\documentclass[12pt]{article}

\usepackage[margin=2.5cm]{geometry}
\usepackage{amsthm,amsmath,amsfonts,amssymb}
\usepackage{graphicx,float}
\usepackage{multirow,setspace}
\usepackage[authoryear]{natbib}
\usepackage{enumerate}
\usepackage{caption}
\usepackage{subcaption}
\usepackage{bm}
\usepackage{color}
\usepackage{xcolor}%
\usepackage{url}
\usepackage[colorlinks,citecolor=blue,urlcolor=blue]{hyperref}
\usepackage[format=plain, labelfont={bf,it}, textfont=it]{caption}
\usepackage{listings}%
\usepackage{lstbayes}
\usepackage{booktabs}
\usepackage{algorithm}%
\usepackage{algorithmicx}%
\usepackage{algpseudocode}%

\theoremstyle{plain}

\newtheorem{theorem}{Theorem}[section]
    
\theoremstyle{remark}
\newtheorem{definition}[theorem]{Definition}

\newcommand{\bbeta}{{\mbox{\boldmath$\beta$}}}

\newcommand{\bepsilon}{{\mbox{\boldmath$\epsilon$}}}

\newcommand{\bgamma}{{\mbox{\boldmath$\gamma$}}}
\newcommand{\blambda}{{\mbox{\boldmath$\lambda$}}}
\newcommand{\bmu}{{\mbox{\boldmath$\mu$}}}

\newcommand{\btheta}{{\mbox{\boldmath$\theta$}}}

\newcommand{\bSigma}{{\mbox{\boldmath$\Sigma$}}}

\newcommand{\bx}{\textbf{x}}

\newcommand{\bc}{\textbf{c}}
\newcommand{\bC}{\textbf{C}}
\newcommand{\bY}{\textbf{Y}}

\newcommand{\one}{\textbf{1}}
\newcommand{\zero}{\textbf{0}}
\newcommand{\up}{\underline{p}}

\title{\vspace{-40pt}Global Test of Functional Effects in Mean Field Variational Bayesian Scalar-on-Function Regression\vspace{-10pt}}
\author{Mark J. Meyer \\ Department of Mathematics and Statistics, Georgetown University \\ Washington, DC USA\\
Chak Kwong (Tommy) Cheng \\ Division of Epidemiology and Biostatistics, University of Illinois Chicago \\ Chicago, IL USA}
\date{}

\begin{document}

	\maketitle
	
\begin{abstract}
Variational regression methods are increasingly popular because they enable efficient estimation in complex statistical models. To characterize their theoretical properties, researchers have established frequentist properties for models fitted using mean-field variational Bayes (MFVB). However, hypothesis testing in this framework remains understudied. We propose a global test for functional effects in scalar-on-function regression. The test combines consistency results for MFVB estimators with an existing test for smoothed semiparametric curves. To implement the procedure, we also develop a coordinate ascent variational inference (CAVI) algorithm for MFVB scalar-on-function regression, rather than relying on existing black-box procedures. Finally, we conduct a simulation study comparing our estimation procedure and global test with fully frequentist and fully Bayesian alternatives in terms of bias, root mean squared error, size, and power.
\end{abstract}

\section{Introduction}

Variational approximations provide a useful class of algorithms for obtaining deterministic approximations to Bayesian posterior distributions. One such approach, mean-field variational Bayes (MFVB), has been studied extensively, and standard statistical models have been developed within this framework \cite{OrmerodWand2010} or \cite{Blei2017}. Research also suggests that MFVB estimates share some properties with frequentist estimators, including asymptotic normality and consistency \citep{Bickel2013,You2014,Wang2019,Zhang2020}. These results suggest that variational methods can have desirable frequentist properties, despite arising from a Bayesian estimation framework.

Notwithstanding this previous work, hypothesis testing remains relatively understudied. Percentile-based intervals for parameters can be constructed from the variational posterior, either directly or through resampling. However, variance underestimation can lead to misleading inference, particularly when the posterior is not normally distributed \citep{Giordano2018}. One response is to fit a fully Bayesian model when inference is required; see \cite{Goldsmith2016}, for example, refit a function-on-scalar regression model when performing inference.  Implementations based on black-box variational inference may also recommend switching to Markov chain Monte Carlo for inference, as in the \text{Stan} variational inference workflow  \citep{RStan2023}. For complex models such as scalar-on-function regression, it would be preferable to perform estimation and inference within a single framework rather than introduce an additional and computationally burdensome method.

Global tests of functional components in regression models are also relatively limited. Much of the existing work focuses on pointwise inference (see, for example, \cite{Goldsmith2013} or \cite{Meyer2015}). Some authors have proposed global tests. \cite{Meyer2015} develop a global Bayesian \(p\)-value for fully Bayesian function-on-function regression, using what they refer to as simultaneous band scores or SimBaS. \cite{Ivanescu2015} note the connection between penalized functional effects and smoothed effects and suggest inference based on approaches similar to those of \cite{Crainiceanu2004} and \cite{Greven2008}. The operating characteristics of these global tests remain underexplored, and global testing in the variational functional-regression setting has not been studied.

Motivated by the consistency results for MFVB estimates discussed by \cite{Wang2019}, we develop a test of functional effects in variational scalar-on-function regression models. The test adapts a frequentist procedure for smoothed semiparametric curves, which we refer to as the ZLS test \citep{Zhang2000}. Because the test statistic requires information that existing black-box implementations do not readily provide, we also develop a coordinate ascent variational inference (CAVI) algorithm for estimating scalar-on-function regression models. Our algorithm takes advantage of the well-known relationship between penalized regression and mixed models.

\section{Scalar-on-Function Regression}
\label{s:sofr}

Let $y_i$ be an outcome of interest for subjects $i = 1, \ldots, n$ and $\bx_i$ be a $p$-dimensional vector of scalar covariates. For the scalar model components, let $\alpha$ denote the intercept and $\bbeta$ the vector of nonfunctional regression coefficients. Assume that the errors are independent and identically distributed as $\epsilon_i \stackrel{iid}{\sim} N(0, \sigma_e^2)$. We place the following priors on $\alpha$, $\bbeta$, and $\sigma_e^2$: $\alpha \sim N(0, \sigma_a^2)$, $\bbeta \sim N(\textbf{0}, \sigma_b^2 I_{p\times p})$, and $\sigma_e^2 \sim IG(a_e, b_e)$. Here, $\textbf{0}$ is $p$-dimensional zero vector, $I_{p\times p}$ is the $p\times p$ identity matrix, and $IG$ denotes the inverse gamma distribution. We fix $\sigma_a^2$ and $\sigma_b^2$ at large values and set $a_e$ and $b_e$ to something small (e.g. 0.01). For $F$ total functional effects, the model is then
\begin{align}
	y_i = \alpha + \bbeta'\bx_i + \sum_{f = 1}^F \int_{t\in\mathcal{T}} w_{if}(t)\gamma_f(t) dt + \epsilon_i, \label{eq:vam}
\end{align}
where $\gamma_f(t)$ is the $f$th functional effect and $w_{1f}(t)$ is the corresponding functional covariate. Although the effects are written over a common domain, the domain may vary with $f$; that is, $t$ and $\mathcal{T}$ mt depend on $f$. For notational convenience, and without loss of generality, we assume that the functions are observed on grids of equal size.

Using a basis expansion, represent each functional effect as $\bgamma_f = \Theta_f\lambda_f$. The basis-expanded version of Equation~\eqref{eq:vam} is then
\begin{align}
	\bY = \alpha\one &+ X\bbeta + \sum_{f=1}^F W_f\Theta_f\blambda_f + \bepsilon, \label{eq:matvam}
\end{align}
where $\bY$, $\bepsilon$, and $\one$ are $n$-dimensional vectors and $X$ is an $n\times p$ matrix of scalar covariates. For each functional components, $\Theta_f$ is a $T\times L$ matrix of B-spline basis functions, $\blambda$ is an $L$-dimensional vector of basis coefficients, and $W_f$ is an $n\times T$ matrix containing the $f$th functional covariate. Although we use B-splines, other lossless or approximately lossless basis expansions could also be used. In our empirical study, models with $L = 8$ knots perform well.

All basis expanded effects require penalization to avoid overfitting. To achieve this, we use the standard mixed-model representation of penalized regression \citep[Chapter 4]{Ruppert2003}. In a frequentist model, the components to be penalized are treated as random effects. In a Bayesian model, unknown parameters are assigned prior distributions, so the corresponding penalty can be induced through an appropriate prior. Specifically, we assign mean-zero shrinkage priors to the elements of $\blambda_f$: $\blambda_f \sim N(\zero, \eta_f\Delta_f)$. Here, $\eta_f$ is a tuning parameter and $\zero$ is a $L$-dimensional zero vector. The penalty matrix is defined as $\Delta_f = \xi \Delta_0 + (1 - \xi)\Delta_2$, where $\Delta_0$ and $\Delta_2$ are the zeroth- and second-derivative penalty matrices for the B-spline basis, respectively. We estimate the tuning parameters to determine the amount of shrinkage. To this end, we assign conditionally conjugate inverse-gamma priors to each $\eta_f$: $\eta_m \sim IG(a_{\eta, m}, b_{\eta, m})$. We set the hyperparameters to small values (0.01 or less).

\section{MFVB Estimation}


We briefly review the MFVB approximation. Let $\btheta$ represent a vector of non-functional parameters and $\bY$ represent a vector of observed data. The posterior distribution of $\btheta$ given the data, $\bY$, is then $p(\btheta | \bY) = p(\bY, \btheta)\big/p(\bY)$ where the marginal likelihood is typically intractable. Posterior approximations must therefore be obtained algorithmically. MFVB utilizes a density transformation approach that relies on the observation that, for an arbitrary density $q$, $p(\bY)$ is bounded below by $\up(\bY; q)$ where $\up(\bY; q) = \exp\left[ \int q(\btheta) \log\left\{ \frac{p(\bY, \btheta)}{q(\btheta)} \right\} d\btheta \right]$. The algorithms used to determine the approximations work by maximizing $\up(\bY; q)$ over a class of tractable densities, $q$. The $\log[\up(\bY; q)]$ is often referred to as the Evidence Lower Bound (ELBO). As the algorithm maximizes the ELBO, the Kullback-Leibler or K-L Divergence between the posterior, $p(\btheta | \bY)$, and its approximation, $q(\btheta)$, will be minimized. This is due to the fact that the ELBO and the K-L Divergence are equal to each other, up to an additive constant.


Partition the parameter vector into $M$ blocks, $\{\btheta_1, \ldots, \btheta_M\}$. MFVB approximates the joint posterior using a density, $q(\btheta)$, formed by taking the product of approximate marginal densities corresponding to each block: $q(\btheta) = \prod_{m = 1}^M q_m(\btheta_m)$. The resulting $q$-densities are similar to the conditional posterior densities from a Gibbs sampler. Optimization proceeds iteratively, and convergence is declared when changes in $\up(\bY; q)$ (or $\log[\up(\bY; q)]$) become minimal. We use coordinate ascent variational inference (CAVI) for this optimization \citep{OrmerodWand2010,Blei2017}.


Combining the model specifications from Secion~\ref{s:sofr}, the resulting $q$ densities for the MFVB approximation to the joint posterior are
	$q(\btheta) \sim N\left[\bmu_{q(\btheta)}, \bSigma_{q(\btheta)}\right],\ 
	q(\sigma_e^2) \sim IG\left[ a_e + \frac{N}{2}, B_{q(\sigma_e^2)} \right],\ 
	 q(\eta_1) \sim IG\left[ a_{\eta,1} + \frac{L}{2}, B_{q(\eta_1)}\right],
	 \ \ldots\ , \text{ and }
	 q(\eta_F) \sim IG\left[ a_{\eta,F} + \frac{L}{2}, B_{q(\eta_F)}\right]$
The subscript-${q(\cdot)}$ indicates the parameter to which the quantity belongs under the mean field approximation. Let $\btheta$ contain all parameters in the mean model let $\bC$ denote the design matrix from Equation~\eqref{eq:matvam}. Then, $\btheta = \left[\begin{array}{cccccccc} \alpha & \bbeta & \blambda_1 & \cdots & \blambda_F \end{array} \right]'$ and  $\bC = \left[ \begin{array}{cccccccc} \one & X & W_1 & \cdots & W_F \end{array} \right]$. Algorithm~\ref{a:vam} describes the estimation procedure, which can accommodate an arbitrary number of functional effects.. Convergence is achieved when changes in $\log[ \up(\bY; q) ]$ become negligible, where
		$\log[ \up(\bY; q) ] = \frac{1}{2}(P + K + L) - \frac{N}{2}\log(2\pi) - \frac{P}{2}\log(\sigma_b^2) + \frac{1}{2}\log\left(|\bSigma_{q(\btheta)}|\right) - \frac{1}{2\sigma_b^2}\left[ {\bmu_{q(\alpha, \bbeta)}}'\bmu_{q(\alpha, \bbeta)} + \text{tr}\left\{\bSigma_{q(\alpha, \bbeta)}\right\} \right]
			- a_e\log(b_e) - \left(a_e + \frac{N}{2}\right)\log(B_{q(\sigma^2)}) + \log\left\{\Gamma\left(a_e + \frac{N}{2}\right)\right\} - \log\left\{\Gamma(a_e)\right\}
			+ \sum_{f = 1}^F \bigg[a_{\eta,f}\log(b_{\eta,f}) - \left(a_{\eta,f} + \frac{L}{2}\right)\log(B_{q(\eta_f)})+ \log\left\{\Gamma\left(a_{\eta,f} + \frac{L}{2}\right)\right\} - \log\{\Gamma(a_{\eta,f})\}\bigg]$.

\begin{algorithm}[h]
\caption{CAVI Algorithm to estimate the components of Equation~\eqref{eq:matvam}. Abbreviations: `bd' stands for block diagonal. } \label{a:vam}
\begin{algorithmic}
\Require $B_{q(\omega_1)}, \ldots, B_{q(\omega_M)}, B_{q(\eta_1)}, \ldots, B_{q(\eta_F)} > 0$, and $\epsilon > 0$, small
	\While{$\Delta \log[ \up(\bY; q) ] > \epsilon$}
		\State $\mathcal{D} \gets \text{bd}\bigg\{ (\sigma_b^2)^{-1} I_{p\times p}, \frac{a_{\eta,1} + \frac{L}{2} }{B_{q(\eta_1)}} \Delta_1, \cdots, \frac{a_{\eta,F} + \frac{L}{2} }{B_{q(\eta_F)}} \Delta_M  \bigg\}$ 
		\State $\bSigma_{q(\btheta)} \gets \left[ \frac{a_e + \frac{N}{2}}{B_{q(\sigma_e^2)}}\bC'\bC + \mathcal{D}\right]^{-1}$
		\State $\bmu_{q(\btheta)} \gets \left( \frac{a_e + \frac{N}{2}}{B_{q(\sigma_e^2)}}\right)\bSigma_{q(\btheta)}\bC'\bY $ 
		\State  $B_{q(\sigma^2)} \gets b_e + \frac{1}{2}\left[ \left\{\bY - \bC\bmu_{q(\btheta)} \right\}'\left\{\bY - \bC\bmu_{q(\btheta)} \right\} + \text{tr}\left\{ \bC'\bC \right\}\bSigma_{q(\btheta)} \right]$
		\Loop $\ f = 1, \ldots, F$
			\State $B_{q(\eta_f)} \gets b_ {\eta,f} + \frac{1}{2} \left[ {\bmu_{q(\blambda_f)}}'\bmu_{q(\blambda_f)} + \text{tr}\left\{ \frac{a_{\eta,f} + \frac{L}{2} }{B_{q(\eta_f)}} \Delta_f \right\} \right]$
		\EndLoop
\EndWhile
\end{algorithmic}
\end{algorithm}

\section{ZLS Test of Functional Effects}

The null hypothesis of interest for the global test is $H_0 : \gamma_f(t) = 0\ \forall\ t \in\mathcal{T} = \{t : t = t_1, \ldots, t_T\}$. We assume that the functional effect is measured on an equally spaced grid, although this assumption can be relaxed. Let $\bSigma_{q(\bgamma_f)}$ denote the principal submatrix of $\bSigma_{q(\btheta)}$ corresponding to $\bgamma_f$.
\begin{definition}\label{d:cf}
    Define the vector $\bc_f(t)$ to be $\bc_f(t) = \left( \frac{a_e + \frac{N}{2}}{B_{q(\sigma_e^2)}}\right)\bSigma_{q(\bgamma_f)}W\Theta_f$.
\end{definition}
\begin{theorem}
    Once Algorithm~\ref{a:vam} converges, and for sufficiently large $N$, $\exists$ a vector $\bc_f(t)$ s.t. $\hat{\gamma}_f(t) = \bc_f(t)'\bY$.
\end{theorem}
\begin{proof}
  Upon convergence, the MFVB estimator of $\gamma_f(t)$ is consistent \citep{Wang2019}. That estimator is given by $\hat{\gamma}_f(t) = \left[\left( \frac{a_e + \frac{N}{2}}{B_{q(\sigma_e^2)}}\right)\bSigma_{q(\bgamma_f)}W\Theta_f\right]'\bY$. By Definition~\ref{d:cf}, $\bc_f(t) = \left( \frac{a_e + \frac{N}{2}}{B_{q(\sigma_e^2)}}\right)\bSigma_{q(\bgamma_f)}W\Theta_f$. Therefore, $\bc_f(t)$ is a vector that satisfies $\hat{\gamma}_f(t) = \bc_f(t)'\bY$.
\end{proof}

For $Q_f = \int_{t_1}^{t_T} \bc_f(t)\bc_f(t)' dt$, the statistic is
	$\kappa^{-1}\bY'Q_f \bY$
which is chi-squared on $\nu = 2e^2/\psi$ degrees of freedom for $\kappa = \psi/(2e)$. Under $H_0$, $e$ and $\psi$ are approximated by $e \approx tr(Q_fV)$ and $\psi \approx 2tr\left\{ (Q_fV)^2 \right\}$ where $V$ is the working covariance of $\bY$ and $tr\{\cdot\}$ denotes the trace \citep{Zhang2000}.

\section{Empirical Study}

We benchmark the ZLS test of variational functional effects against two alternative procedures: an established Frequentist based test using penalized functional regression (PFR) \citep{refund2022} and a fully Bayesian implementation (FB) based on 2000 posterior samples following a warmup of equal length \citep{Ziren2025}. The global test for the PFR follows \cite{Crainiceanu2004}. To calculate a global test for the FB approach, we find the Global Bayesian $p$-value (GBPV) using the SimBaS procedure described by \cite{Meyer2015}.

We generate data from the one- and two-function models given by
	$y_i = \alpha + \int_{t\in\mathcal{T}} w_{i}(t)\gamma(t) dt + \epsilon_i$ and $y_i = \alpha + \int_{t\in\mathcal{T}} w_{i1}(t)\gamma_1(t) dt + \int_{t\in\mathcal{T}} w_{i2}(t)\gamma_2(t) dt + \epsilon_i$,
with errors distributed as $\epsilon_i \stackrel{iid}{\sim} N(0, 1)$ in both models. We generate the functional covariates using a Gaussian process with an AR(1) correlation structure. For both one- and two-function models, we evaluate bias, root mean squared error ($\sqrt{MSE}$), size, and power when $T = 100$, $\alpha = 0$, and $N = 100, 200, 500$ and 1000. The true curves are $\gamma_1(t) = \xi\sin\left\{\pi\left(4t - 1\right) \right\}\left(t + \frac{1222}{10000} \right)$ (Sine) and $\gamma_2(t) = \frac{\xi }{4}\sqrt{\frac{100}{2\pi}} \exp\left\{ -\frac{100}{2}\left(t - \frac{1}{4} \right)^2 \right\}$ (Normal) where the parameter $\xi$ varies from 0 to 1 by 0.1 to evaluate power. For each model configuration, sample size $N$, and scaling effect $\xi$, we generate 5000 simulated datasets. All simulations were run using 50 cluster cores on the Lakeshore High Performance Computing cluster at the University of Illinois Chicago.

\begin{table}
	\centering
	\caption{Average computation time ($s$) by model.\label{t:time}}
	\begin{tabular}{rrcccc}
		\hline
		Number of &  \multirow{2}{*}{Model}  & \multicolumn{4}{c}{Sample Size} \\
		\cline{3-6}
        Curves & & {100} & {200} & {500} & {1000} \\
		\hline
      One & MFVB & 0.025 & 0.039 & 0.044 & 0.132 \\
          & PFR & 0.001 & 0.002 & 0.006 & 0.011 \\
          & FB & 0.091 & 0.162 & 0.394 & 0.830 \\
         \cline{2-6}
    Two  & MFVB &  0.054 & 0.080 & 0.124 & 0.283 \\
          & PFR & 0.002 & 0.005 & 0.011 & 0.023 \\
          & FB & 0.147 & 0.282 & 0.716 & 1.387 \\
		\hline 
	\end{tabular}
\end{table}

\begin{table}
	\centering
	\caption{Size using the nominal $\alpha$, Monte Carlo error $\approx 0.006$.\label{t:size}}
	\begin{tabular}{rrrcccc}
		\hline
		Number of & \multirow{2}{*}{Curve Type} & \multirow{2}{*}{Model}  & \multicolumn{4}{c}{Sample Size} \\
		\cline{4-7}
        Curves & & & {100} & {200} & {500} & {1000} \\
		\hline
      One & Sine & MFVB & \bf 0.052 & \bf 0.048 & 0.057 & \bf 0.050 \\
         & & PFR & 0.054 & \bf 0.052 & \bf 0.054 & 0.053 \\
         & & FB & 0.015 & 0.016 & 0.018 & 0.017 \\
         \cline{3-7}
         & Normal & MFVB & \bf 0.052 & \bf 0.048 & 0.057 & \bf 0.050 \\
         & & PFR & 0.054 & \bf 0.052 & \bf 0.054 & 0.053 \\
         & & FB & 0.016 & 0.015 & 0.017 & 0.017 \\
         \cline{2-7}
    Two & Sine & MFVB & 0.053 & 0.047 & 0.055 & \bf 0.046 \\
         & & PFR & \bf 0.052 & \bf 0.052 & \bf 0.053 & \bf 0.054 \\
         & & FB & 0.014 & 0.012 & 0.018 & 0.016 \\
         \cline{3-7}
         & Normal & MFVB & 0.053 & 0.058 & 0.048 & \bf 0.052 \\
         & & PFR & \bf 0.050 & \bf 0.054 & \bf 0.051 & 0.059 \\
         & & FB & 0.015 & 0.013 & 0.016 & 0.018 \\
		\hline 
	\end{tabular}
\end{table}

Tables~\ref{t:time} and~\ref{t:size} summarize the average computation times and empirical test sizes, respectively, by curve type, number of functional effects, and sample size. Table~\ref{t:time} reports the average time, in seconds, required by each method to complete its estimation procedure. As expected, PFR is the fastest method. However, MFVB is typically an order of magnitude faster than the FB approach. Table~\ref{t:size} shows that MFVB has empirical size similar to that of the fully frequentist, PFR-based global test; both are close to the nominal level of $\alpha = 0.05$. In contrast, the FB approach's global Bayesian $p$-value is substantially more conservative under the null, consistent with previous findings \citep{Meyer2015}.

\begin{figure}
    \centering
    \includegraphics[height=1.55in]{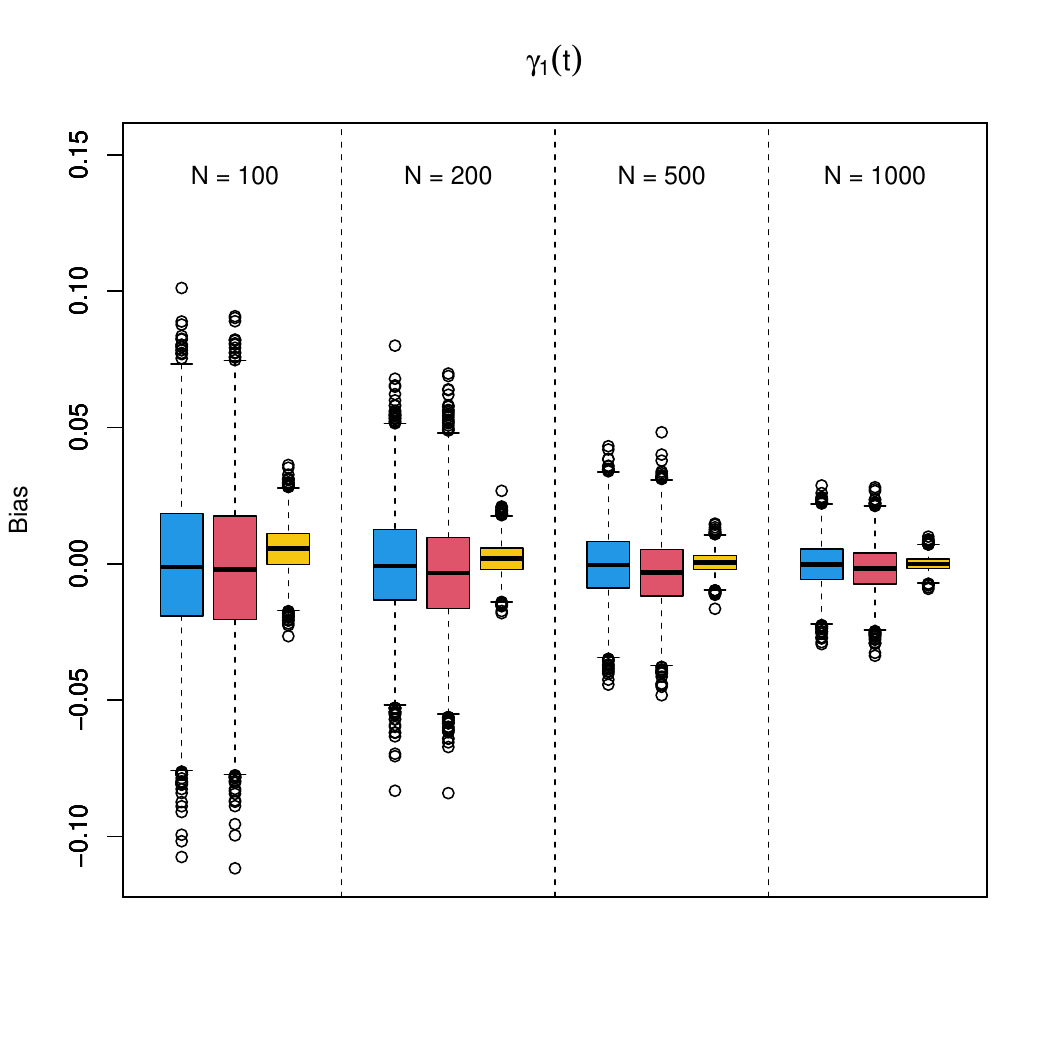}
    \includegraphics[height=1.55in]{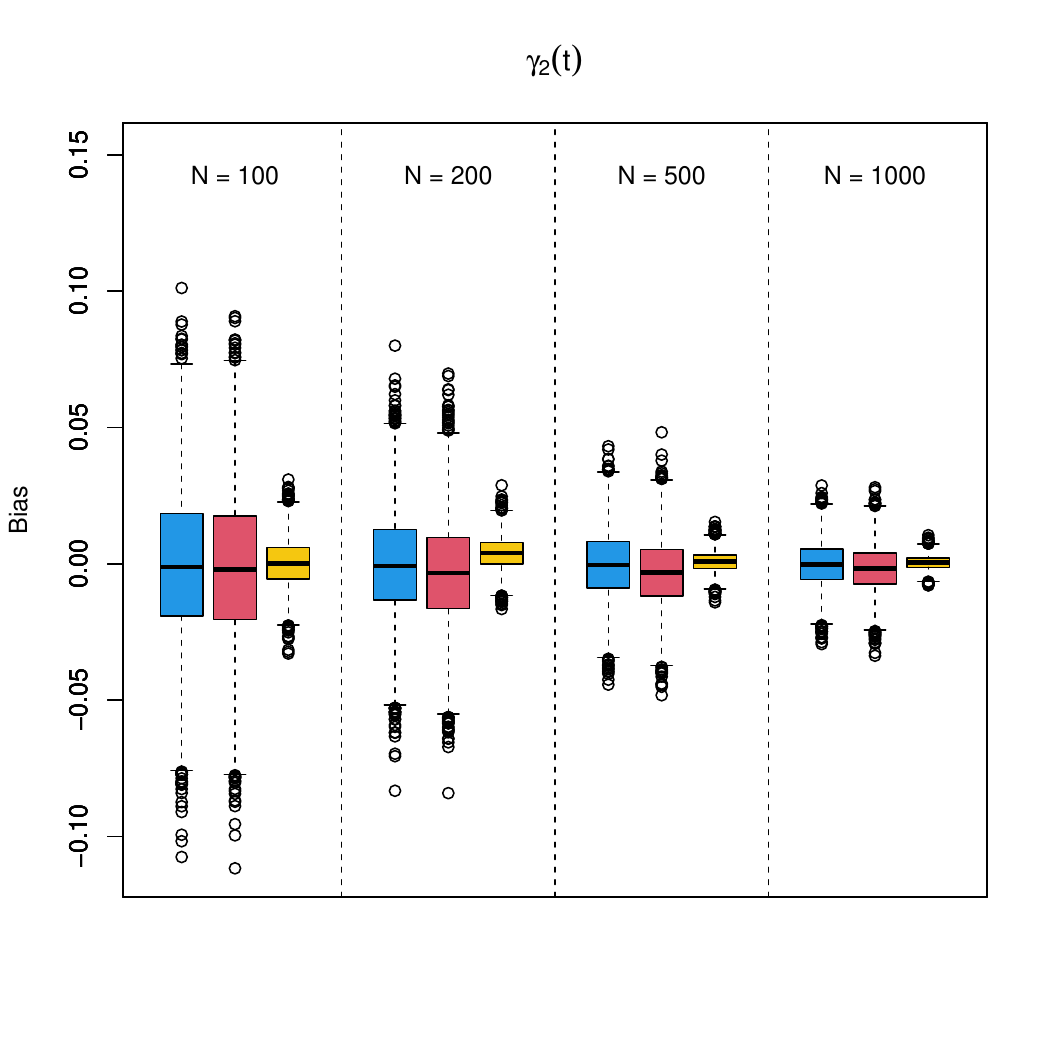}
    \includegraphics[height=1.55in]{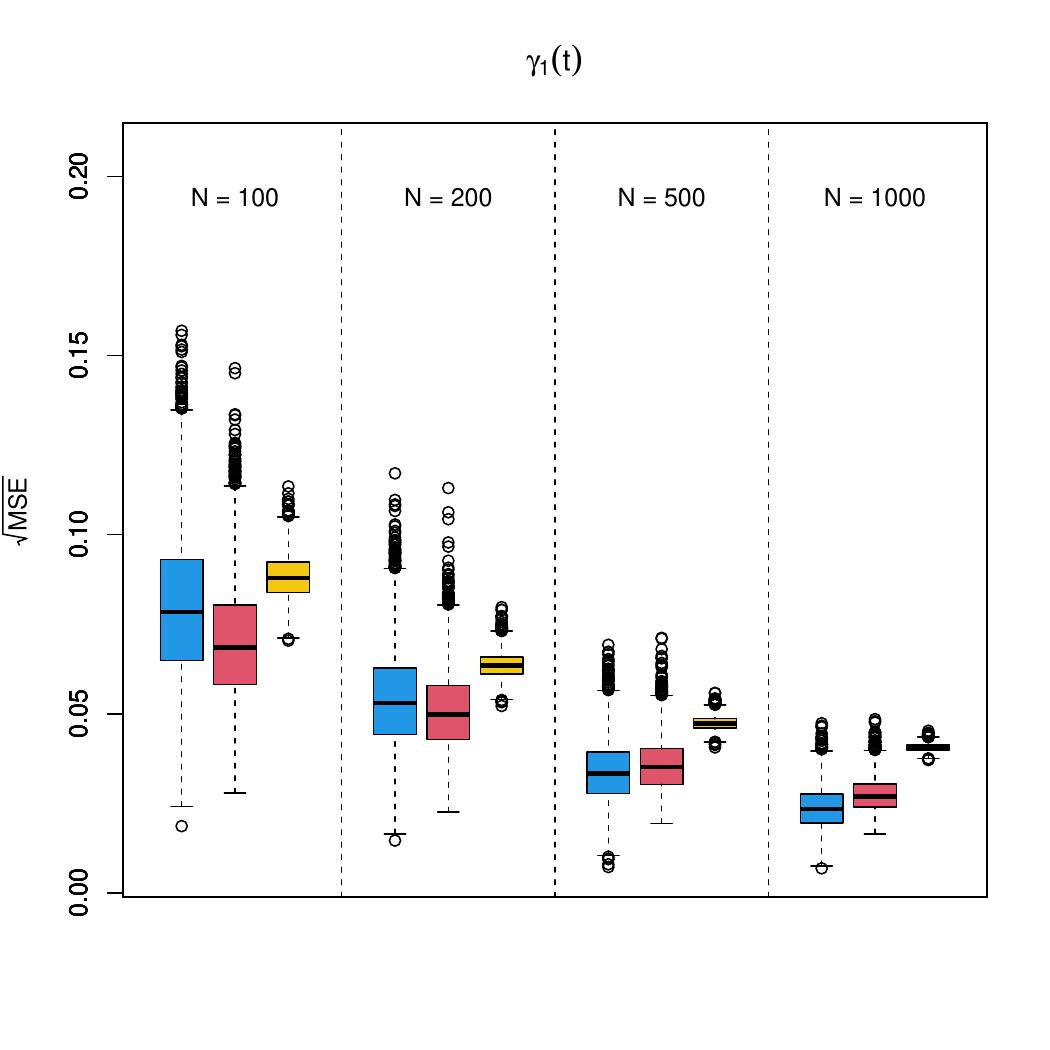}
    \includegraphics[height=1.55in]{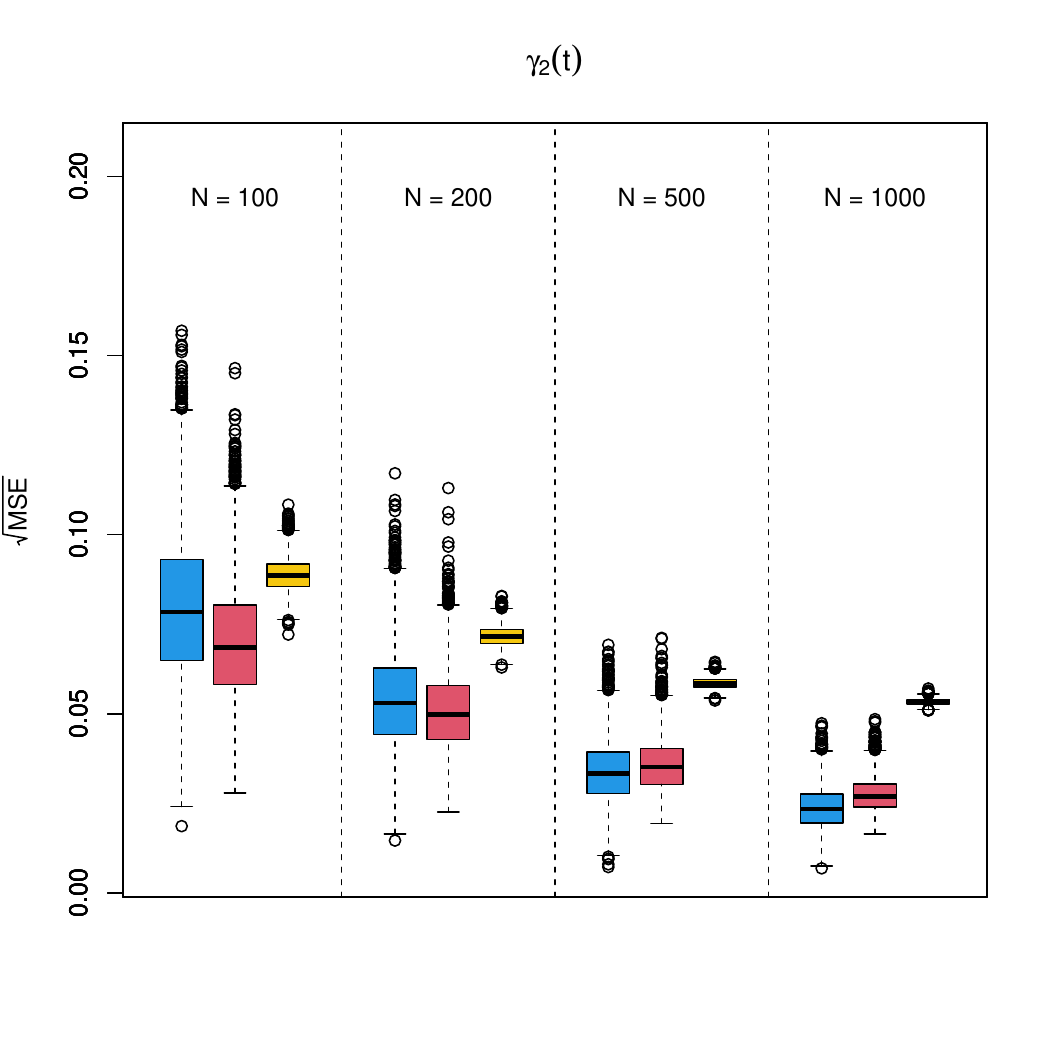}
    \includegraphics[height=1.55in]{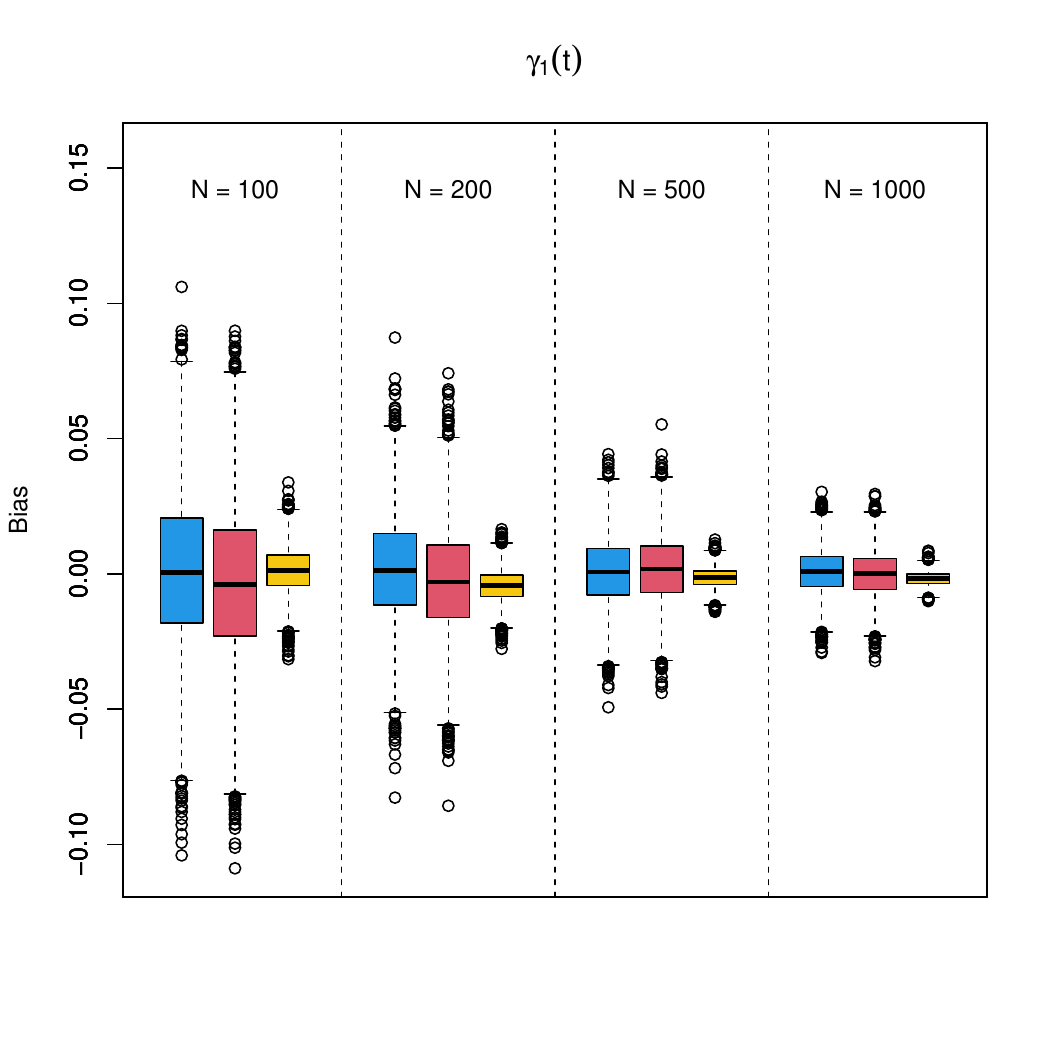}
    \includegraphics[height=1.55in]{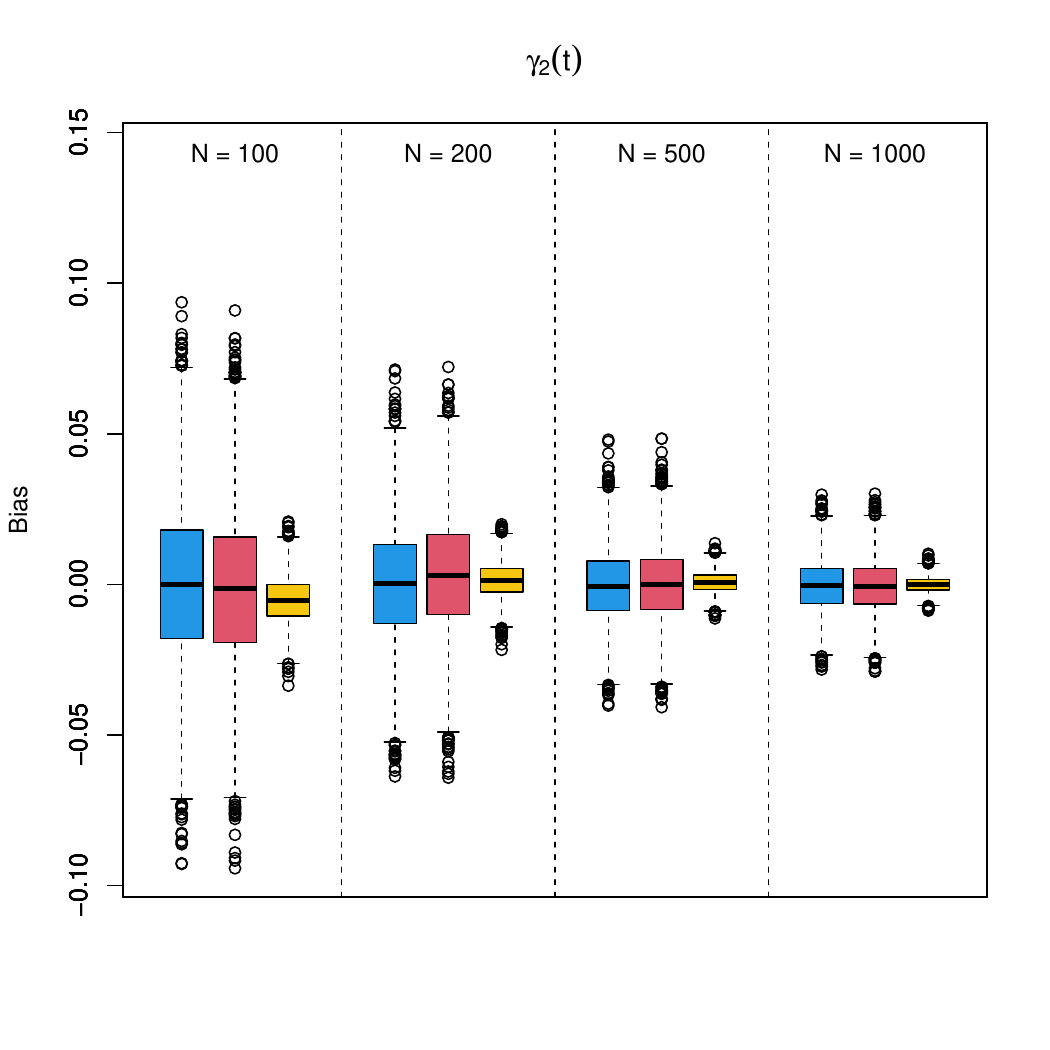}
    \includegraphics[height=1.55in]{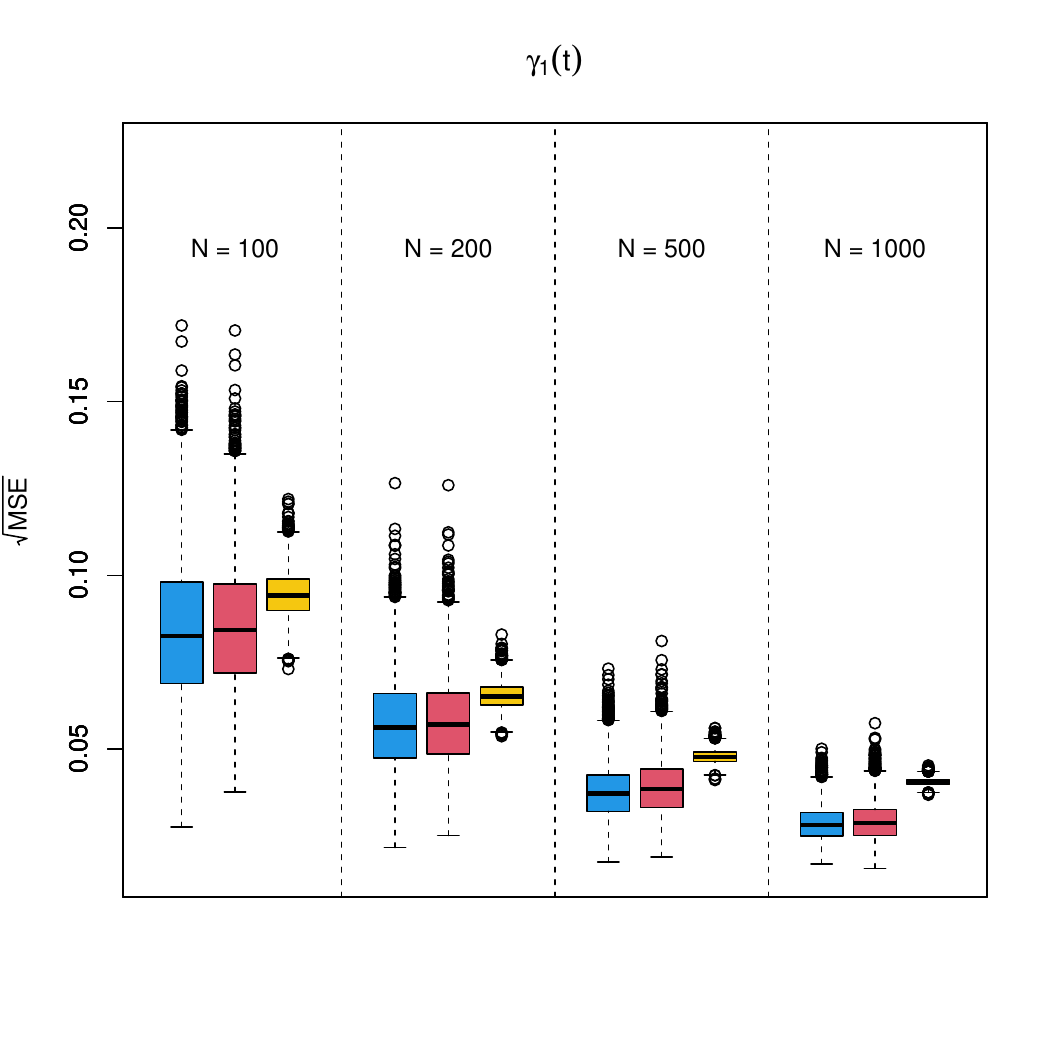}
    \includegraphics[height=1.55in]{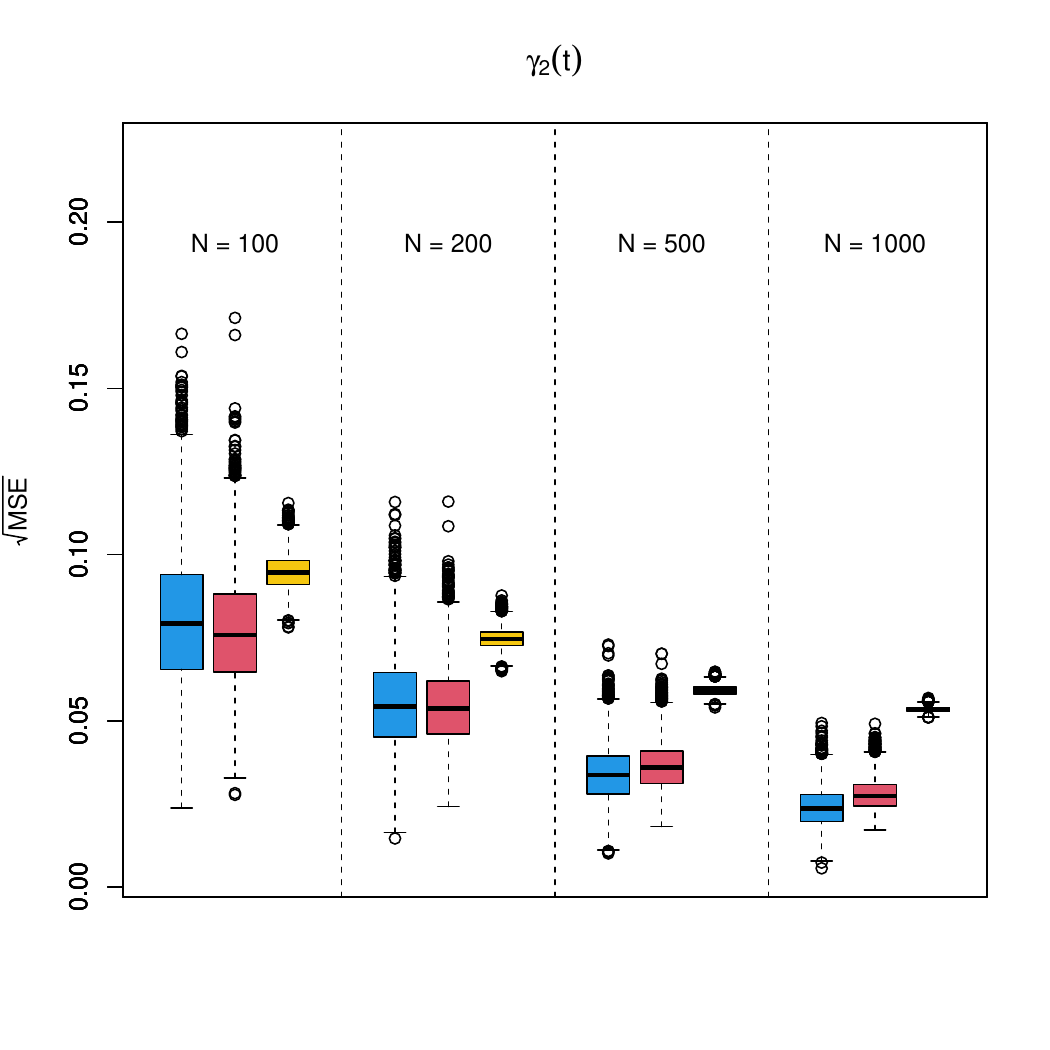}
    \caption{One function models (top row) and two function models (bottom row) from the simulated datasets with two functions. The left two columns present the bias, the right two columns display  $\sqrt{MSE}$. Effect alternates by column: sine ($\gamma_1(t)$), normal ($\gamma_2(t)$), sine, and normal. MFVB is in blue, PFR is in red, and FB is in gold.}
    \label{f:biasRmse}
\end{figure}

Figure~\ref{f:biasRmse} summarizes the bias and $\sqrt{MSE}$. For all graphs, the results for MFVB are in blue, for PFR are in red, and for FB are in gold. The first two columns contain boxplots of bias for all 5,000 simulated data sets, grouped by sample size; the last two columns contain boxplots of  $\sqrt{MSE}$. MFVB and PFR generally have biases closer to zero than FB, with MFVB tending to have the smallest bias. For all methods, the variability in bias decreases as the sample size increases. $\sqrt{MSE}$ also decreases with sample size, as does its variability across simulated data sets. For $N\geq 200$, MFVB generally has lower  $\sqrt{MSE}$, whereas FB tends to have the highest $\sqrt{MSE}$. Neither the number of functional predictors nor the type of true curve (sine versus normal) appears to have a substantial effect on bias or  $\sqrt{MSE}$.


\begin{figure}
    \centering
    \includegraphics[height=1.55in]{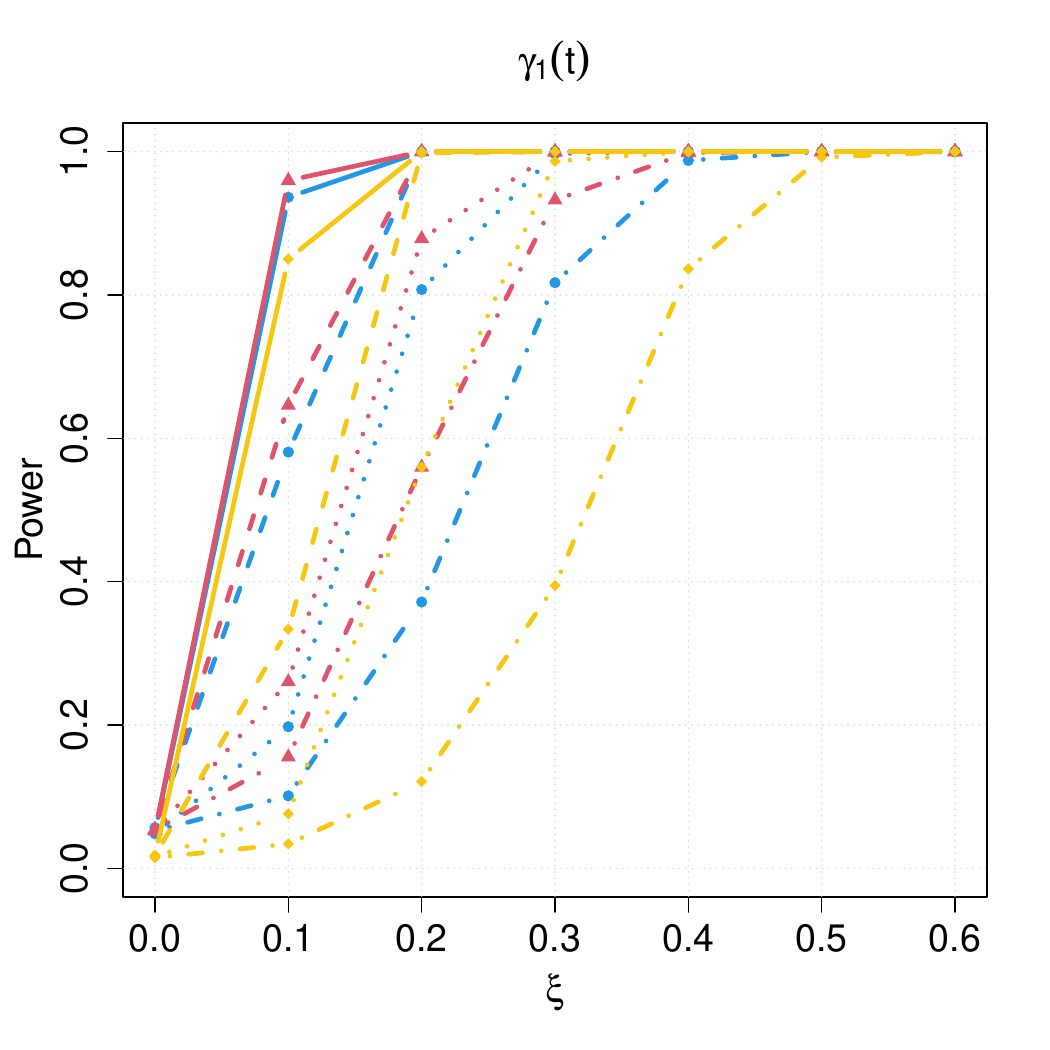}
    \includegraphics[height=1.55in]{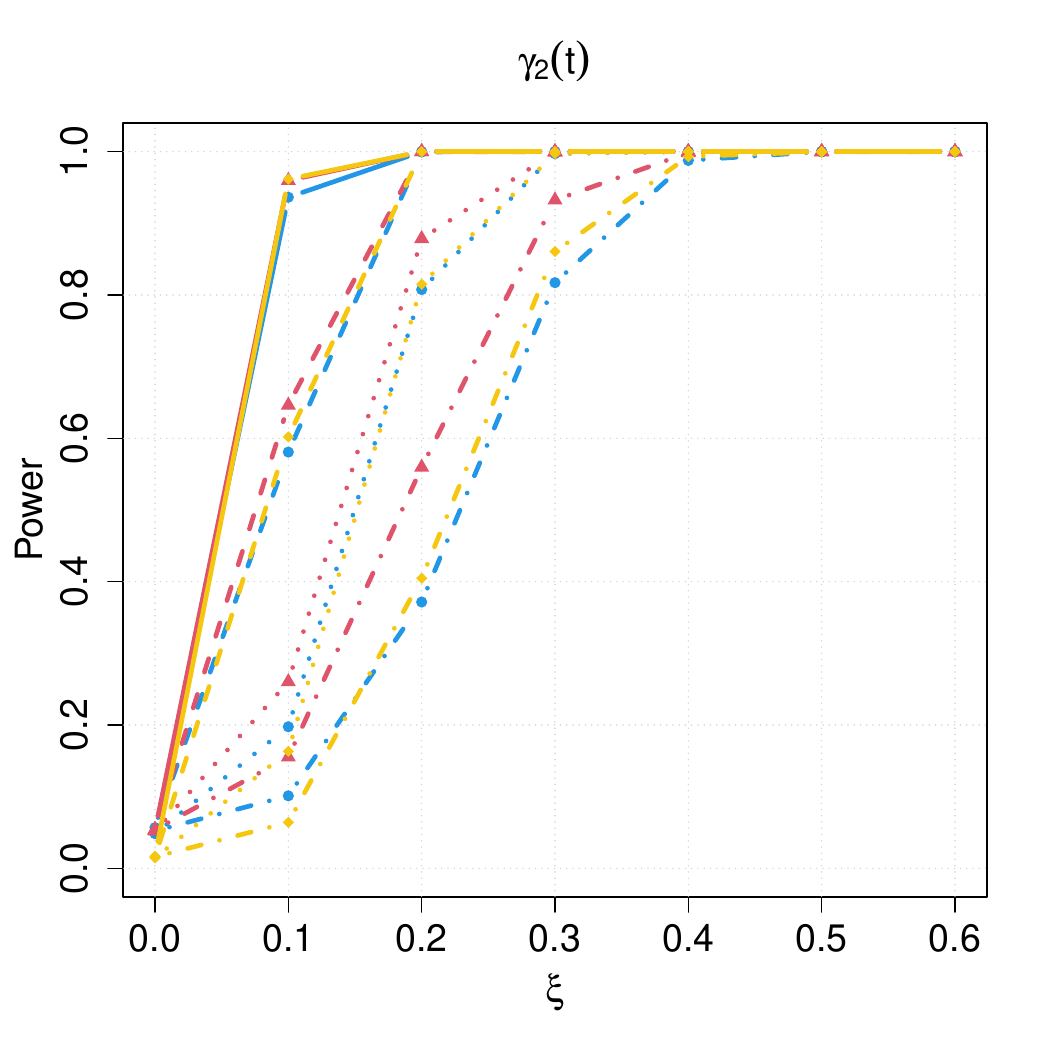}
    \includegraphics[height=1.55in]{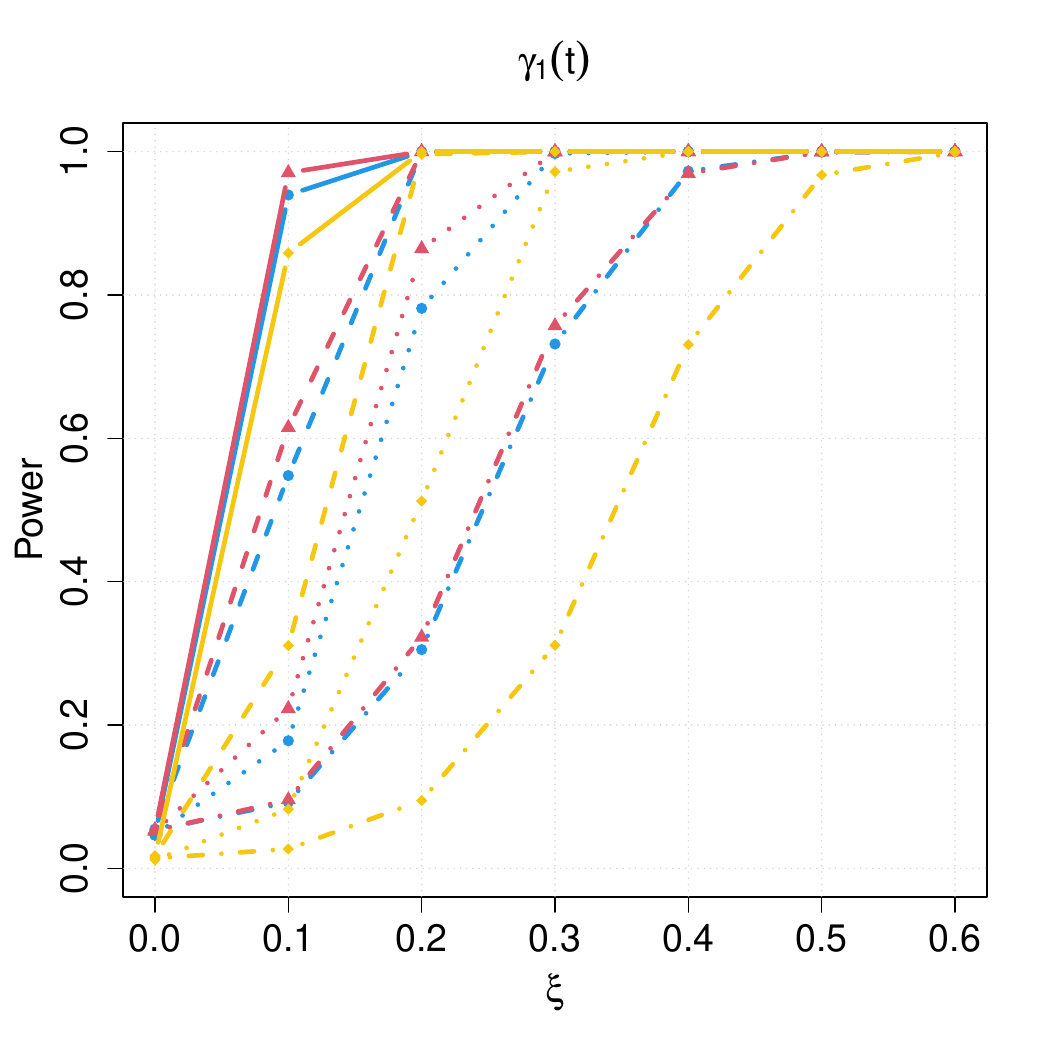}
    \includegraphics[height=1.55in]{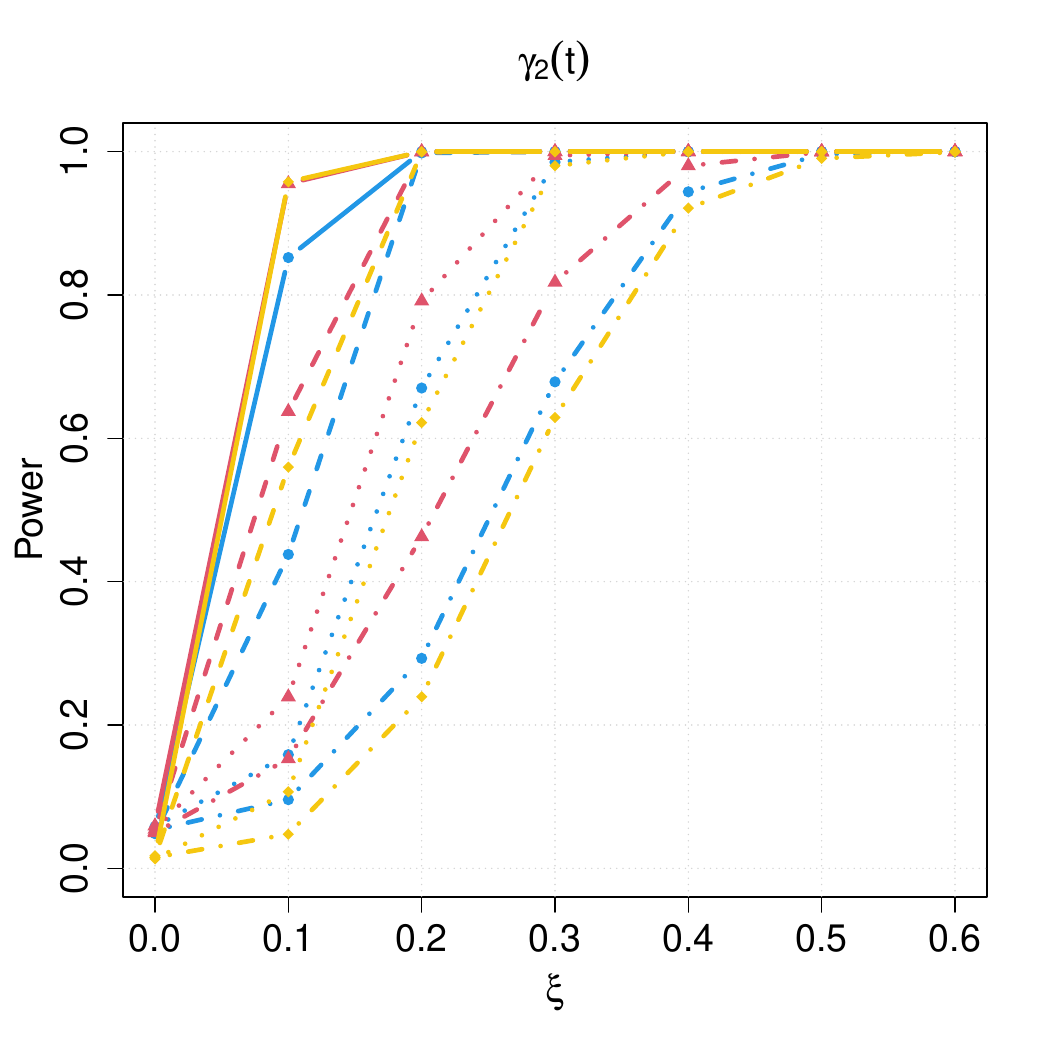}
    \caption{Power from the simulated datasets with one function (left two graphs) and two functions (right two graphs). Effect alternates by column: sine ($\gamma_1(t)$), normal ($\gamma_2(t)$), sine, and normal. MFVB results are in blue with circles, PFR results are in red with triangles, and FB results are in gold with diamonds. Dotted-dashed curves correspond to a sample size of $N = 100$, dotted to $N = 200$, dashed to $N = 500$, and solid to $N = 1000$.}
    \label{f:power}
\end{figure}

Figure~\ref{f:power} shows the estimated power curves. The results for MFVB are in blue with circles, for PFR are in red with triangles, and for FB are in gold with diamonds. The dotted-dashed curves correspond to a sample size of $N = 100$, dotted to $N = 200$, dashed to $N = 500$, and solid to $N = 1000$. The two left panels correspond to models with one functional predictor and the two right panels correspond to models with two functional predictors. Regardless of the method, power increases with sample size. PFR generally has the highest power, although MFVB approaches its performance as the sample size increases. At smaller sample sizes, the FB approach based on global Bayesian $p$-values tends to have the lowest power.


\section{Discussion}

Given the ability of variational methods to quickly estimate complex regression models, it is useful to have inference procedures built within the same framework. We propose just such an approach: a frequentist-based test statistic for functional effects in variational scalar-on-function regression, relying on the consistency of the MFVB estimator. Our simulations show that the proposed test has favorable empirical properties and performs similarly to an established fully frequentist test. The MFVB approach additionally tends to have the lowest $\sqrt{MSE}$, bias near zero, and empirical size close to the nominal level. It also carries substantially lower computational costs than the FB approach.

Our aim is not to replace fully Bayesian or fully frequentist methods. Practitioners committed to either framework may have little reason to change approaches. Rather, the proposed ZLS test provides a way to conduct inference within the MFVB scalar-on-function regression framework. Because the test relies on consistency of the MFVB estimator, we do not recommend using it with small samples. Code implementing all three methods is available at \url{https://github.com/markjmeyer/VSOF}.

\bibliographystyle{abbrvnat} 
\bibliography{fullbib.bib}

\end{document}